 \documentclass[conference,10pt]{IEEEtran}        
 \usepackage[T1]{fontenc}
 \usepackage[
top    = 0.7 in,
bottom = 1.1 in,
left   = 0.625 in,
right  = 0.625 in]{geometry}

\usepackage[usenames,dvipsnames]{xcolor}
\usepackage{amsfonts}
\usepackage[sort]{cite}
\usepackage{amsmath,amsthm} 
\usepackage{amssymb}  
\usepackage{graphicx}
\usepackage{epsfig}
\usepackage{epstopdf}
\usepackage{mathtools}
\usepackage{etoolbox}
\usepackage{lipsum}
\usepackage{caption}
\usepackage{comment}
\usepackage{soul}
\usepackage{verbatim}
\usepackage{subfig} 
\usepackage{algorithm,algorithmic}

\newcommand{\T}{^{\mbox{\tiny T}}}

\newtheorem{theorem}{\bf Theorem}

\newtheorem{remark}{\bf Remark}
\newtheorem{problem}{\bf Problem}
\newtheorem{assumption}{\bf Assumption}

\newtheorem{proposition}{\bf Proposition}

\newcommand{\be}{\begin{equation}}
\newcommand{\ee}{\end{equation}}
\newcommand{\bea}{\begin{eqnarray}}
\newcommand{\eea}{\end{eqnarray}}
\newcommand{\bes}{\begin{eqnarray*}}
\newcommand{\ees}{\end{eqnarray*}}
\newcommand{\ba}{\begin{align}}
\newcommand{\ea}{\end{align}}

\newcommand{\bfi}{\begin{figure}}
\newcommand{\bfit}{\begin{figure}[t]}
\newcommand{\bfib}{\begin{figure}[b]}
\newcommand{\bfih}{\begin{figure}[h]}
\newcommand{\bfip}{\begin{figure}[p]}
\newcommand{\efi}{\end{figure}}

\newcommand{\bi}{\begin{itemize}}
\newcommand{\ei}{\end{itemize}}
\newcommand{\ben}{\begin{enumerate}}
\newcommand{\een}{\end{enumerate}}

\newcommand{\bp}{\begin{problem}}
\newcommand{\ep}{\end{problem}}

\newenvironment{list4}{
  \begin{list}{$\bullet$}{
      \setlength{\itemsep}{0.05cm}
      \setlength{\labelsep}{0.2cm}
      \setlength{\labelwidth}{0.3cm}
      \setlength{\parsep}{0in}
      \setlength{\parskip}{0in}
      \setlength{\topsep}{0in}
      \setlength{\partopsep}{0in}
      \setlength{\leftmargin}{0.17in}}}
      {\end{list}}

\IEEEoverridecommandlockouts
\ifCLASSINFOpdf
\else
\fi
\usepackage[cmintegrals]{newtxmath}
\begin{document}



\title{Distributed Average Consensus in Wireless Multi-Agent Systems with Over-the-Air Aggregation}

\author{
    \IEEEauthorblockN{Themistoklis Charalambous\IEEEauthorrefmark{1}, Zheng Chen\IEEEauthorrefmark{2}, and Christoforos N. Hadjicostis\IEEEauthorrefmark{1}} \thanks{The work of T. Charalambous is partly funded by MINERVA, which received funding from the European Research Council (ERC) under the European Union's Horizon 2022 research and innovation programme (Grant Agreement No. 101044629).}
    \IEEEauthorblockA{\IEEEauthorrefmark{1}Department of Electrical and Computer Engineering, University of Cyprus, Nicosia, Cyprus
    \\Emails: \{charalambous.themistoklis, hadjicostis.christoforos\}@ucy.ac.cy}
    \IEEEauthorblockA{\IEEEauthorrefmark{2}Department of Electrical Engineering, Link\"{o}ping University, Link\"{o}ping, Sweden
    \\Email: zheng.chen@liu.se}
}


\maketitle
%
%
%
%
\begin{abstract}
In this paper, we address the average consensus problem of multi-agent systems over wireless networks. We propose a distributed average consensus algorithm by invoking the concept of over-the-air aggregation, which exploits the signal superposition property of wireless multiple-access channels. The proposed algorithm deploys a modified version of the well-known Ratio Consensus algorithm with an additional normalization step for compensating for the arbitrary channel coefficients. We show that, when the noise level at the receivers is negligible, the algorithm converges asymptotically to the average for time-invariant and time-varying channels. Numerical simulations corroborate the validity of our results.
\end{abstract}

%
%
%
%
\section{Introduction}\label{sec:intro}

The unprecedented interconnection of multi-agent systems we experience today necessitates efficient network management. 
One problem of particular interest over multi-agent systems is reaching \textit{consensus} in a distributed fashion, i.e., reaching agreement among distributed agents towards a common decision. 
\textit{Average consensus} constitutes a special case of consensus in which agents are initially endowed with numerical states and aim to calculate the average among these initial states. 
%
In classical consensus approaches, each node acquires separately the states of all the neighboring agents, whereas only a function (e.g., a linear combination) of the states is needed~\cite{2003:jadbabaie_coordination}. 

Acquiring the individual states over wireless networks increases the communication overhead and cost considerably and, in addition, requires advanced coordination in order to avoid interference and packet collisions. 
A key characteristic of the wireless channel is signal superposition, meaning that a receiving node can receive signals from multiple transmitting nodes simultaneously. In traditional wireless communication systems, concurrent transmissions from multiple sources are usually undesired due to the negative effect of interference on the decodability of the intended signal. On the other hand, signal superposition in multiple access channels (MACs) can also be exploited to perform function computation efficiently ~\cite{nazer2007computation}.
Inspired by this concept, over-the-air aggregation has recently attracted wide attention in distributed settings where the goal is to aggregate information from a large number of source nodes.
Specifically, all nodes operate in the same frequency-time blocks and transmit analog signals carrying the information of their states, thus forming a waveform at the receiver side that corresponds to the aggregated information.
 



There are limited existing works in the literature considering distributed consensus with over-the-air aggregation as the communication approach. Some earlier works have explored this direction without considering the fading effect~\cite{Boche:2012WiOpt, 2012:WCNC}. 
In~\cite{2021:Max-Consensus,2022:Molinari_Max-Consensus}, over-the-air aggregation over fading channels 
for \emph{$\max$-consensus} is considered.\footnote{The $\max$-consensus algorithm is a distributed algorithm for computing the maximum value among the initial states of agents in multi-agent systems.} Furthermore, in~\cite{2018:Molinari}, a \emph{weighted} average consensus protocol is proposed under the same conditions. Even though consensus is reached, the protocol does not necessarily reach the average. Recently, a stochastic approximation-based (diminishing step size) distributed consensus protocol has been proposed in~\cite{yang2024distributed}, which guarantees convergence to the average in the mean square sense, for the case with noisy channels and complex-valued channel coefficients without phase alignment. However, the channel and noise statistics are assumed to be known in these works. A similar method with non-coherent processing was proposed in~\cite{michelusi2022non} for the distributed optimization setting.

In this paper, we propose the Over-the-Air Ratio Consensus algorithm with which we can achieve the exact average consensus in both time-invariant and time-varying channels, when the noise level at the receivers is negligible~\cite{2021:Max-Consensus,2022:Molinari_Max-Consensus,2018:Molinari}, but without requiring any knowledge about the channels. Our proposed method relies on the inherent and fundamental property of channel reciprocity~\cite{2004:Reciprocity_Smith} and employs a modified version of the well-celebrated Ratio Consensus algorithm~\cite{2010:christoforos}.

%
%

\section{Network Model and Preliminaries}\label{sec:prelim}
 \subsection{Notation}

 The set of real (integer) numbers is denoted by $\mathbb{R}$ ($\mathbb{Z}$) and the set of nonnegative numbers (integers) is denoted by $\mathbb{R}_{+}$ ($\mathbb{Z}_{+}$).  $\mathbb{R}^n_+$ denotes the nonnegative orthant of the $n$-dimensional real space $\mathbb{R}^n$. Vectors are denoted by small letters whereas matrices are denoted by capital letters. The transpose of a matrix $A$ is denoted by $A\T$. For $A\in \mathbb{R}^{n\times n}$, $A_{ij}$ denotes the entry at row $i$ and column $j$.  
 By $\mathbf{1}$, we denote the all-ones vector, and by $I$ we denote the identity matrix (of appropriate dimensions).

\subsection{Network model}

We consider a wireless network consisting of $n\in\mathbb{N}$ spatially distributed nodes,  captured by a \emph{directed} graph $\mathcal{G}=(\mathcal{N}, \mathcal{E})$, where $\mathcal{N}=\{v_1, \cdots, v_n\}$ is the set of nodes (representing the $n$ agents) and $\mathcal{E} \subseteq \mathcal{N} \times \mathcal{N}$ is the set of edges (representing the communication links between agents). The total number of edges in the network is denoted by $m=|\mathcal{E}|$. A directed edge $\varepsilon_{ji} \triangleq (v_j, v_i) \in \mathcal{E}$, where $v_j, v_i \in \mathcal{N}$, indicates that node $v_j$ can receive information from node $v_i$, i.e., $v_i \rightarrow v_j$. The nodes that transmit information to node $v_j$ directly are called in-neighbors of node $v_j$, and belong to the set $\mathcal{N}_{j}^{-}=\{v_i \in \mathcal{N} | \varepsilon_{ji} \in \mathcal{E}\}$. The number of nodes in the in-neighborhood is called in-degree and it is represented by the cardinality of the set of in-neighbors, $d_{j}^{-} = |\mathcal{N}_{j}^{-}|$. The nodes that receive information from node $v_j$ directly are called out-neighbors of node $v_j$, and belong to the set $\mathcal{N}_{j}^{+}=\{v_l \in \mathcal{N} | \varepsilon_{lj} \in \mathcal{E}\}$. The number of nodes in the out-neighborhood is called out-degree and it is represented by the cardinality of the set of out-neighbors, $d_{j}^{+}= |\mathcal{N}_{j}^{+}|$. Note that when self-loops are included in digraph $\mathcal{G}$, the number of in-coming links of node $v_j$ is ($d_j^- +1$) and similarly the number of its out-going links is ($d_j^+ +1$).

\subsection{Ratio Consensus}

In \cite{2010:christoforos}, an algorithm is suggested that solves the average consensus problem in a directed graph in which each node $v_j$ distributively sets the weights on its self-link and outgoing-links to be $p_{lj}=\frac{1}{1+d_j^+} \ \forall (v_l,v_j)\in\mathcal{E}$, so that the resulting weight matrix $P=[p_{lj}]$ (with $P_{lj}=p_{lj}$) is column stochastic, but not necessarily row stochastic. Asymptotic average consensus is reached by using this weight matrix to run two iterations with appropriately chosen initial conditions. 

\begin{proposition}[\hspace{-0.001cm}\cite{2010:christoforos}]~\label{lemma_christoforos}
Consider a strongly connected digraph $\mathcal{G}(\mathcal{N}, \mathcal{E})$. Let $y_j[k]$ and $x_j[k]$ (for all $v_j \in \mathcal{N}$ and $k=0,1,2,\ldots$) be the result of the iterations
\begin{subequations}\label{eq:1}
\begin{align}
y_j[k+1]=p_{jj} y_j[k]+ \sum_{v_i \in \mathcal{N}^{-}_j} p_{ji} y_i[k] \; , \label{subeq:1} \\
x_j[k+1]=p_{jj} x_j[k]+ \sum_{v_i \in \mathcal{N}^{-}_j} p_{ji} x_i[k] \; , \label{subeq:2}
\end{align}
\end{subequations}
where $p_{lj} = \frac{1}{1 + d_j^+}$ for $v_l \in \mathcal{N}_j^+ \cup \{ v_j \}$ (zero otherwise), and the initial conditions are ${y_j[0] =S_j}\in \mathbb{R}$ and  $x_j[0]=1$, for $v_j\in\mathcal{N}$. Then, the solution to the average consensus problem can be asymptotically obtained as
$$
\displaystyle \lim_{k\rightarrow \infty} \mu_j[k]=\frac{\sum_{v_j\in \mathcal{N}} y_j[0]}{\sum_{v_j\in \mathcal{N}} x_j[0]} {=\frac{\sum_{v_j\in \mathcal{N}} S_j}{|\mathcal{N}|}} \; , \forall v_j \in \mathcal{N} \; ,
$$
where
$
\displaystyle \mu_j[k]=\frac{y_j[k]}{x_j[k]} \; .
$   \hfill $\lrcorner$
\end{proposition}
\begin{remark}
Proposition~\ref{lemma_christoforos} states a distributed algorithm with which the exact average is \emph{asymptotically} reached, even if the directed graph is not balanced.  \hfill $\lrcorner$
\end{remark}

\begin{remark}
The algorithm described in Proposition~\ref{lemma_christoforos} is for a specific choice of weights on each link that assumes that each node knows its out-degree. Note, however, that the algorithm works for any set of weights that adhere to the graph structure and form a primitive column stochastic weight matrix.     \hfill $\lrcorner$
\end{remark}

\subsection{Communication Mechanism}

At every iteration of the ratio consensus algorithm, each node needs to send its state information to its out-neighbors and receive information from its in-neighbors. To facilitate the communication procedure, we make the following assumption.


\begin{assumption}\label{assumption:4}
The agents are equipped with full-duplex transceivers, such that they can transmit and receive simultaneously. Perfect self-interference cancellation is also assumed so that the transmitted signal from an agent will not affect the received signals from its neighbors.
\end{assumption}

\begin{remark}
Assumption~\ref{assumption:4} can be relaxed if either each time step is divided into two-time slots (one for receiving and one for transmitting) or each agent decides whether to transmit or receive at a given time step (randomly or deterministically). 
\end{remark}

\subsection{Channel model}
For transmitting data over wireless links, random variation of signal attenuation is unavoidable, which is captured by the fading coefficients.
Let $s_i[k]$ denote the transmitted data symbol from source node $v_i$ in the $k$-th iteration. The received signal at each node 
$v_i \in \mathcal{N}$ is given by
\begin{align}
    r_i[k] = \sum_{v_j \in  \mathcal{N}_i^{-}[k]}h_{ij}[k] s_j[k] + \nu_i[k],
    \label{eq:received-signal}
\end{align}
where  $h_{ij}[k]\sim \mathcal{N}(0,\Sigma_{ij})$ denotes the channel coefficient between agent $v_i$ and $v_j$, and $\nu_i[k]\sim \mathcal{N}(0,\Sigma_{i})$ represents the additive white Gaussian noise.

\begin{assumption}\label{assumption:1}
We assume that the communication channel coefficients are modeled as positive real-valued random numbers.
\end{assumption}
In communication systems, channel coefficient is typically modeled as a complex number, reflecting the effect of the channel on the signal in terms of both amplitude attenuation and phase shift. In this work, we assume that the fading coefficient  $h_{ij}$ is represented as a positive real-valued random variable. This assumption is valid for specific types of modulation schemes where phase information is not essential, such as amplitude modulation with envelope detection. A similar assumption is present in \cite{2021:Max-Consensus}.


\begin{assumption}\label{assumption:2}
The noise level is assumed to be negligible.
\end{assumption}
In practice, noise is an unavoidable feature in communication systems. In this study, we start from the simplest case where noise is negligible as compared to the signal. The effect of noise will be investigated in an extended version of this paper.

\begin{assumption}\label{assumption:3}
Channel reciprocity holds in the entire network, i.e., $h_{ij}=h_{ji}, \forall (v_i,v_j) \in\mathcal{E}$.
\end{assumption}
The reciprocity principle is based on the property that electromagnetic waves traveling between two antenna locations will undergo the same physical perturbations (e.g., reflection, refraction, diffraction) in both directions. Hence, if the link operates on the same frequency band in both directions, the impulse response of the channel observed between any two antennas should be the same regardless of the direction \cite{2004:Reciprocity_Smith}.

%
%
%
%
\section{Over-the-Air Ratio Consensus}

In this section, we propose the \textit{Over-the-Air Ratio Consensus} algorithm, under which the multi-agent system can achieve the average asymptotically. We consider two variants of the same algorithm: one with time-invariant channels (i.e., the channel state remains the same throughout the operation of the algorithm) and another one with time-varying channels (i.e., the channels may vary at every time step).

\subsection{Time-invariant channels}

For the case with the time-invariant channel, the main operation steps of the algorithm (herein called the TIC-Over-the-Air Ratio Consensus) are described as follows. 
\begin{list4}
    \item[1.] Each agent $v_j$ maintains an initial state $\tilde{y}_j[0]=S_j$. 
    The initial value of the auxiliary variable $\tilde{x}_j[k]$ is also set to $\tilde{x}_j[0] = 1$. At the initialization stage, agent $v_j$ obtains $\sigma_j[0] = \sum_{v_i\in \mathcal{N}_{j}^{-}} h_{ji}$ (which is the superposition of all incoming channels) by having all agents transmit $1$. Then, as initial states for the algorithm, each agent sets $y_j[0]=\frac{\tilde{y}_j[0]}{\sigma_j[0]}$ and $x_j[0]=\frac{\tilde{x}_j[0]}{\sigma_j[0]}$.
    \item[2.] At every time step $k=0,1,2, \ldots$, each agent broadcasts $y_j[k]$ and $x_j[k]$ in two different time slots (the time slots are assumed to be within the same time step).
    \item[3.] At every time step $k=0, 1,2, \ldots$, agent $v_j$ receives the corresponding aggregated signals from the in-neighbors:
    \begin{subequations} 
    \begin{align}
        \tilde{y}_j[k+1] &=\sum_{v_i\in \mathcal{N}_{j}^{-}} h_{ji} y_i[k] , \label{eq:received signals-a} \\
        \tilde{x}_j[k+1] &=\sum_{v_i\in \mathcal{N}_{j}^{-}} h_{ji} x_i[k] . \label{eq:received signals-b}
    \end{align}
    \end{subequations}
    Then, it prepares the next signals to be transmitted by normalizing the received signals as follows:
    \begin{subequations} 
    \begin{align}
    {y}_j[k+1] &= \frac{\tilde{y}_j[k+1]}{\sigma_j[0]} , \label{eq:normalization-a} \\
    {x}_j[k+1] &=\frac{\tilde{x}_j[k+1]}{\sigma_j[0]} . \label{eq:normalization-b}
    \end{align}
    \end{subequations}
    \item[4.] The output of the algorithm is the ratio $\mu_j[k]\triangleq {y}_j[k]/{x}_j[k]$, which is shown in Theorem~\ref{theorem:1} to converge to the exact average, i.e., 
    $\lim_{k\to\infty}\mu_j[k] \equiv \lim_{k\to\infty} \frac{{y}_j[k]}{{x}_j[k]} = \bar{\mu}$,
    where $\bar{\mu}=\frac{\sum_{v_i\in \mathcal{N}} S_i}{|\mathcal{N}|}$.
\end{list4}

The procedure of the TIC-Over-the-Air Ratio Consensus algorithm is summarized in Algorithm~\ref{algorithm:1}. 
\begin{algorithm}[!]
\caption{TIC-Over-the-Air Ratio Consensus}
\textbf{Input:} A strongly connected digraph $\mathcal{G}(\mathcal{N}, \mathcal{E})$ with $n=|\mathcal{N}|$ nodes and $m=|\mathcal{E}|$ edges.\\[0.1cm]
\textbf{Initialization:} Each node $v_j \in \mathcal{N}$, that follows the protocol, sets its initial conditions to $\tilde{y}_j[0] = S_j$ and $\tilde{x}_j[0] =1$. Before starting the iteration each node computes its normalization values $\sigma_j[0] = \sum_{v_i\in \mathcal{N}_{j}^{-}} h_{ji}$, and its initial values to be transmitted, 
$y_j[0]=\frac{\tilde{y}_j[0]}{\sigma_j[0]}$ and $x_j[0]=\frac{\tilde{x}_j[0]}{\sigma_j[0]}$.

\vspace{0.2cm}

\noindent For $k=0,1,2, \ldots$, each node $v_j \in \mathcal{N}$ does the following: \\
\noindent 1) Obtains the aggregated signals $\tilde{y}_j[k+1]$ and $\tilde{x}_j[k+1]$ according to Eqs~\eqref{eq:received signals-a}--\eqref{eq:received signals-b}. \\
\noindent 2) Computes the next signals to be transmitted ${y}_j[k+1]$ and ${x}_j[k+1]$ according to Eqs~\eqref{eq:normalization-a}--\eqref{eq:normalization-b}.

\vspace{0.2cm}

\textbf{Output:} The ratio $\mu_j[k]$.
\label{algorithm:1}
\end{algorithm}

\begin{theorem}\label{theorem:1}
Consider a strongly connected digraph $\mathcal{G}(\mathcal{N}, \mathcal{E})$. Let $y_j[k]$ and $x_j[k]$ (for all $v_j \in \mathcal{N}$ and $k=0,1,2,\ldots$) be the result of the iterations ~\eqref{eq:received signals-a}--\eqref{eq:received signals-b} and~\eqref{eq:normalization-a}--\eqref{eq:normalization-b}, and the initial conditions be ${y_j[0] =S_j}\in \mathbb{R}$ and  $x_j[0]=1$. Then, the solution to the average consensus problem can be asymptotically obtained as
$$
\displaystyle \lim_{k\rightarrow \infty} \mu_j[k]=\frac{\sum_{v_j\in \mathcal{N}} y_j[0]}{\sum_{v_j\in \mathcal{N}} x_j[0]} {=\frac{\sum_{v_j\in \mathcal{N}} S_j}{|\mathcal{N}|}} \; , \forall v_j \in \mathcal{N} \; ,
$$
where
$
\displaystyle \mu_j[k]=\frac{y_j[k]}{x_j[k]} \; .
$   \hfill $\lrcorner$
\end{theorem}

\begin{proof}
Eq.~\eqref{eq:received signals-a} in matrix form is
    $\tilde{y}[k+1] = Hy[k]$,
where
\begin{align*}
\tilde{y}[k]&=\begin{bmatrix}
\tilde{y}_1[k] & \ldots & \tilde{y}_{n}[k] 
\end{bmatrix}^T, \\ 
y[k]&=\begin{bmatrix}
y_1[k] & \ldots & y_{n}[k] 
\end{bmatrix}^T, \\
H &=\begin{bmatrix}
h_{11} & h_{12} & \cdots & h_{1n} \\
h_{21} & h_{22} & \cdots & h_{2n} \\
\vdots & \vdots & \ddots & \vdots \\
h_{n1} & h_{n2} & \cdots & h_{nn}
\end{bmatrix}.    
\end{align*}
Eq.~\eqref{eq:normalization-a} in matrix form is 
$y[k+1] = \Sigma\tilde{y}[k+1]$,
where 
\begin{align*}
\Sigma=\begin{bmatrix}
\sigma^{-1}_{1}[0] & 0 & \cdots & 0 \\
0 & \sigma^{-1}_{2}[0] & \cdots & 0 \\
\vdots & \vdots & \ddots & \vdots \\
0 & 0 & \cdots & \sigma^{-1}_{n}[0]
\end{bmatrix}.    
\end{align*}
Therefore, 
$\tilde{y}[k+1]=H\Sigma\tilde{y}[k] =(H\Sigma)^{k+1}\tilde{y}[0]=\bar{H}^{k+1} \tilde{y}[0]$.
Due to the reciprocity principle, the sum of rows of matrix $H$ is equal to the sum of its columns, and as a consequence, $\bar{H}\triangleq H\Sigma$ is a column stochastic matrix. Similarly, for the auxiliary variable $x[k]$. Therefore, for each node $v_j\in\mathcal{N}$,
\begin{align*}
\lim_{k\to\infty}\mu_j[k]&=\lim_{k\to\infty}\frac{y_j[k]}{x_j[k]}=\lim_{k\to\infty}\frac{\tilde{y}_j[k]}{\tilde{x}_j[k]}=\lim_{k\to\infty}\frac{\bar{H}^{k}(j,:)\tilde{y}[0]}{\bar{H}^{k}(j,:)\tilde{x}[0]} \\ 
&=\frac{\mathbf{v}(j)\mathbf{1}^T \tilde{y}[0]}{\mathbf{v}(j)\mathbf{1}^T \tilde{x}[0]} =\frac{\sum_{v_j\in \mathcal{N}} y_j[0]}{|\mathcal{N}|} =  \frac{\sum_{v_j\in \mathcal{N}} S_j}{|\mathcal{N}|},
\end{align*}
where $\mathbf{v}(j)$ is the j$^{\mathrm{th}}$ entry of the right eigenvector of $\bar{H}$ that corresponds to eigenvalue $1$ (i.e., $\bar{H}\mathbf{v}=\mathbf{v}$). Note that $\mathbf{v}$ or $\mathbf{v}(j)$ need not be known to node $v_j$. Also, note that $\mathbf{v}$ is unique if the underlying graph is strongly connected.
%
\end{proof}

\subsection{Time-varying channels}

For the case with time-varying channel, the main operation steps of the algorithm (herein called the TVC-Over-the-Air Ratio Consensus) are described as follows. 
\begin{list4}
    \item[1.] The initialization step is the same as in the previous case with time-invariant channels.
    \item[2.] When the algorithm starts, at time step $k=0,1,2,\ldots$, each agent broadcasts $y_j[t]$, $x_j[t]$, and $w_j[t]=1$ in three different time slots (the time slots are assumed to be within the same time step).
    \item[3.] At every time step $k=0,1,2, \ldots$, agent $v_j$ receives the corresponding aggregated signals from the in-neighbors:
    \begin{subequations} 
    \begin{align}
        \tilde{y}_j[k+1] &=\sum_{v_i\in \mathcal{N}_{j}^{-}} h_{ji}[k] y_i[k]\; , \label{eq:received signals-a-tv} \\
        \tilde{x}_j[k+1] &=\sum_{v_i\in \mathcal{N}_{j}^{-}} h_{ji}[k] x_i[k]\; , \label{eq:received signals-b-tv} \\
        \sigma_j[k] &=\sum_{v_i\in \mathcal{N}_{j}^{-}} h_{ji}[k] w_i[k] \; , \label{eq:received signals-c} 
    \end{align}
    \end{subequations}
    where $w_i[k]=1$.
    Then, it prepares the next signals to be transmitted by normalizing the received signals as follows:
    \begin{subequations} 
    \begin{align}
    {y}_j[k+1] &= \frac{\tilde{y}_j[k+1]}{\sigma_j[k]} , \label{eq:normalization-a-tv} \\
    {x}_j[k+1] &=\frac{\tilde{x}_j[k+1]}{\sigma_j[k]} .  \label{eq:normalization-b-tv}
    \end{align}
    \end{subequations}
    \item[4.] The output of the algorithm is the ratio $\mu_j[k]\triangleq {y}_j[k]/{x}_j[k]$, which is shown in Theorem~\ref{theorem:2} to converge to the exact average, i.e., 
    \begin{align*}
    \lim_{k\to\infty}\mu_j[k] \equiv \lim_{k\to\infty} \frac{{y}_j[k]}{{x}_j[k]} = \bar{\mu},
    \end{align*}
    where $\bar{\mu}=\frac{\sum_{v_i\in \mathcal{N}} S_i}{|\mathcal{N}|}$.
\end{list4}


The procedure of the TVC-Over-the-Air Ratio Consensus algorithm is summarized in Algorithm~\ref{algorithm:2}.
\begin{algorithm}[!]
\caption{TVC-Over-the-Air Ratio Consensus}
\textcolor{black}{\textbf{Input:} Consider a setting with $\mathcal{N} = \{ v_1, \cdots, v_n \}$ nodes with time-varying communication channels, and assume that $(\epsilon, B)$-connectivity holds.}

\textbf{Initialization:} Each node $v_j \in \mathcal{N}$, that follows the protocol, sets its initial conditions to $\tilde{y}_j[0] = S_j$ and $\tilde{x}_j[0] =1$. At each step, each node computes its normalization values $\sigma_j[k] = \sum_{v_i\in \mathcal{N}_{j}^{-}} h_{ji}[k]$, and its values to be transmitted are updated as follows:
$y_j[k]=\frac{\tilde{y}_j[k]}{\sigma_j[k]}$ and $x_j[k]=\frac{\tilde{x}_j[k]}{\sigma_j[k]}$.

\vspace{0.2cm}

\noindent For $k=0,1,2, \ldots$, each node $v_j \in \mathcal{N}$ does the following: \\
\noindent 1) Obtains the aggregated signals $\tilde{y}_j[k+1]$ and $\tilde{x}_j[k+1]$ according to Eqs~\eqref{eq:received signals-a-tv}--\eqref{eq:received signals-b-tv}. \\
\noindent 2) Computes the next signals to be transmitted ${y}_j[k+1]$ and ${x}_j[k+1]$ according to Eqs~\eqref{eq:normalization-a-tv}--\eqref{eq:normalization-b-tv}.

\vspace{0.2cm}

\textbf{Output:} The ratio $\mu_j[k]$.
\label{algorithm:2}
\end{algorithm}

\textcolor{black}{
$(\epsilon, B)$-Connectivity: Consider a setting with $\mathcal{N} = \{ v_1, \cdots, v_n \}$ nodes with time-varying communication channels among them. At step $k$, any two nodes, say $v_i$ and $v_j$, share a communication channel with time-varying coefficients $h_{ij}[k] = h_{ji}[k]$. For $\epsilon > 0$, let $\mathcal{E}[k] = \{ (v_j, v_i) \: | \; h_{ji}[k] > \epsilon \}$. Considering the sequence of digraphs $\mathcal{G}[k] = (\mathcal{N}, \mathcal{E}[k])$, we require that there exists integer $B$ such that, for all $k \in \mathbb{N}$, the joint digraph $\mathcal{G}[kk] \cup \mathcal{G}[kB+1] \cup \cdots \cup \mathcal{G}[kB+B-1]$ is strongly connected. This is an adaptation of the so-called $B$-connectivity condition (see, e.g., [8], [13]) with dependence on the parameter $\epsilon$.
}

\begin{theorem}\label{theorem:2}
\textcolor{black}{Consider a setting with $\mathcal{N} = \{ v_1, \cdots, v_n \}$ nodes with time-varying communication channels, and assume that $(\epsilon, B)$-connectivity holds.}
Let $y_j[k]$ and $x_j[k]$ (for all $v_j \in \mathcal{N}$ and $k\in\mathbb{N}$) be the result of the iterations ~\eqref{eq:received signals-a-tv}--\eqref{eq:received signals-b-tv} and~\eqref{eq:normalization-a-tv}--\eqref{eq:normalization-b-tv}, and the initial conditions are ${y_j[0] =S_j}\in \mathbb{R}$ and  $x_j[0]=1$. Then, the solution to the average consensus problem can be asymptotically obtained as
$$
\displaystyle \lim_{k\rightarrow \infty} \mu_j[k]=\frac{\sum_{v_j\in \mathcal{N}} y_j[0]}{\sum_{v_j\in \mathcal{N}} x_j[0]} {=\frac{\sum_{v_j\in \mathcal{N}} S_j}{|\mathcal{N}|}} \; , \forall v_j \in \mathcal{N} \; ,
$$
where
$
\displaystyle \mu_j[k]={y_j[k]}/{x_j[k]} \; .
$   
\hfill $\lrcorner$
\end{theorem}

\begin{proof}
Eq.~\eqref{eq:received signals-a-tv} in matrix form is
$\tilde{y}[k+1] = H[k]y[k]$,
where 
\begin{align*}
H[k]=\begin{bmatrix}
h_{11}[k] & h_{12}[k] & \cdots & h_{1n}[k] \\
h_{21}[k] & h_{22}[k] & \cdots & h_{2n}[k] \\
\vdots & \vdots & \ddots & \vdots \\
h_{n1}[k] & h_{n2}[k] & \cdots & h_{nn}[k]
\end{bmatrix}.
\end{align*}
\begin{remark}
For the link $\varepsilon_{lj}\notin \mathcal{E}$, $h_{jl}=h_{lj}<\epsilon$. This accounts for the case for which some channels are in deep fading.  
\end{remark}

\noindent Eq.~\eqref{eq:normalization-a-tv} in matrix form is 
$y[k+1] = \Sigma[k+1]\tilde{y}[k+1]$,
where 
\begin{align*}
\Sigma[k]=\begin{bmatrix}
\sigma^{-1}_{1}[k] & 0 & \cdots & 0 \\
0 & \sigma^{-1}_{2}[k] & \cdots & 0 \\
\vdots & \vdots & \ddots & \vdots \\
0 & 0 & \cdots & \sigma^{-1}_{n}[k]
\end{bmatrix}.    
\end{align*}
Therefore, 
\begin{align*}
\tilde{y}[k+1]&=H[k]\Sigma[k]\tilde{y}[k] =
\underbrace{H[k]\Sigma[k]\ldots H[0]\Sigma[0]}_{\triangleq B_{k:0}} \tilde{y}[0]
\end{align*}
Due to the reciprocity principle, the sum of rows of matrix $H[k]$ is equal to the sum of its columns, and as a consequence, $\bar{H}[k]\triangleq H[k]\Sigma[k]$ is a column stochastic matrix. Similarly, for the auxiliary variable $x[k]$. Therefore, for each node $v_j\in\mathcal{N}$,
\begin{align*}
\lim_{k\to\infty}\mu_j[k]&=\lim_{k\to\infty}\frac{y_j[k]}{x_j[k]}=\lim_{k\to\infty}\frac{B_{k:0}(j,:)\tilde{y}[0]}{B_{k:0}(j,:)\tilde{x}[0]} \\ 
\stackrel{(a)}{=}&\frac{\mathbf{v}_{B_{k:0}}(j)\mathbf{1}^T \tilde{y}[0]}{\mathbf{v}_{B_{k:0}}(j)\mathbf{1}^T \tilde{x}[0]} =\frac{\sum_{v_j\in \mathcal{N}} y_j[0]}{|\mathcal{N}|} =  \frac{\sum_{v_j\in \mathcal{N}} S_j}{|\mathcal{N}|},
\end{align*}
where $\mathbf{v}_{B_{k:0}}(j)$ is the j$^{\mathrm{th}}$ entry of the right eigenvector of $B_{k:0}$ that corresponds to eigenvalue $1$ (i.e., $B_{k:0}\mathbf{v}_{B_{k:0}}=\mathbf{v}_{B_{k:0}}$). Note that $\mathbf{v}_{B_{k:0}}$ or $\mathbf{v}_{B_{k:0}}(j)$ need not be known to node $v_j$. Also, note that $\mathbf{v}_{B_{k:0}}$ is unique if the underlying graph of the matrices comprising $B_{k:0}$ is jointly strongly connected; for the derivation of equality $(a)$ we make similar arguments to those of the theorem by Wolfowitz \cite{1963:Wolfowitz}.
The detailed proof will be provided in an extended version of this paper.
\end{proof}

%
%
%
%
\section{Numerical Examples}\label{sec:examples}

We consider a network of $n=10$ agents, with randomly selected initial values that have an average equal to $1$. 


First, we consider the case for which the channels are time-invariant. The multi-agent system reaches average consensus, as shown in Fig.~\ref{example1}.

\begin{figure}[ht!]
\centering
\includegraphics[width=0.95\columnwidth]{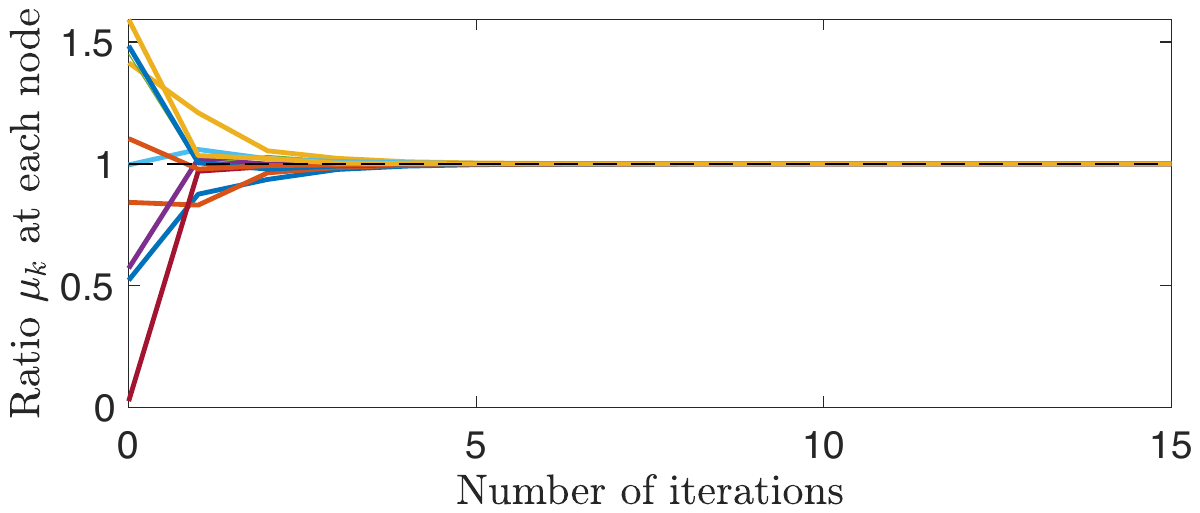}
\caption{Time-invariant channel Over-the-Air Ratio Consensus for a graph of $10$ nodes\vspace{-0.3cm}.}
\label{example1}
\end{figure}


Next, we consider time-varying channels, which are, nevertheless, coherent in every time step. The multi-agent system reaches average consensus, as shown in Fig.~\ref{example2} (bottom), even though the state $y[k]$ (Fig.~\ref{example2}, top)  and auxiliary variable $x[k]$ (Fig.~\ref{example2}, middle) do not.
\begin{figure}[ht!]
\centering
\vspace{-0.13cm}
\includegraphics[width=0.95\columnwidth]{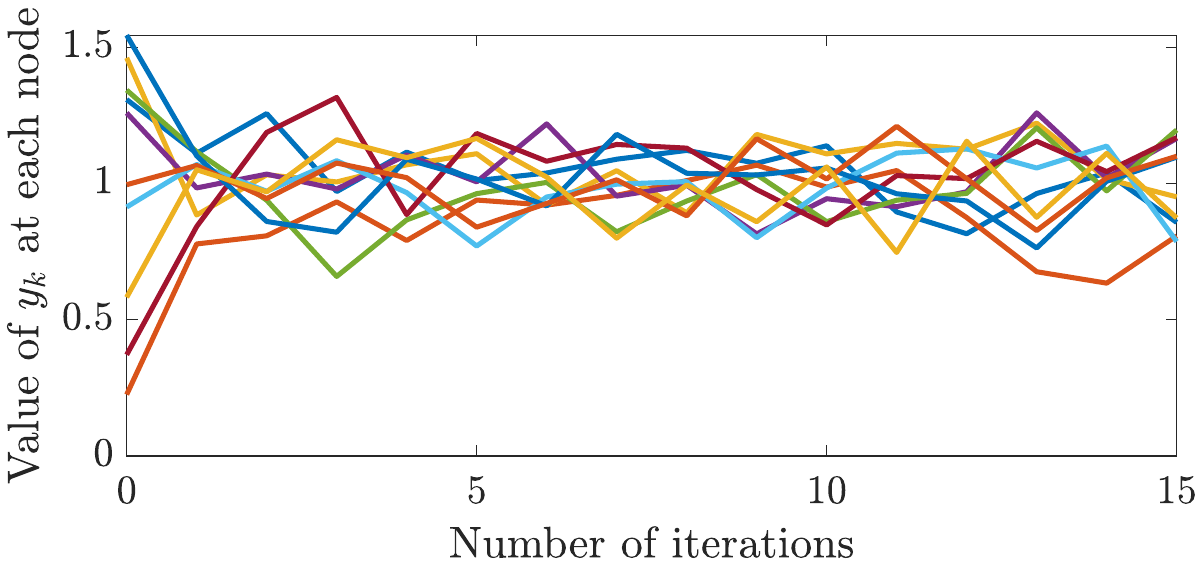}
\vspace{-0.03cm}
\includegraphics[width=0.95\columnwidth]{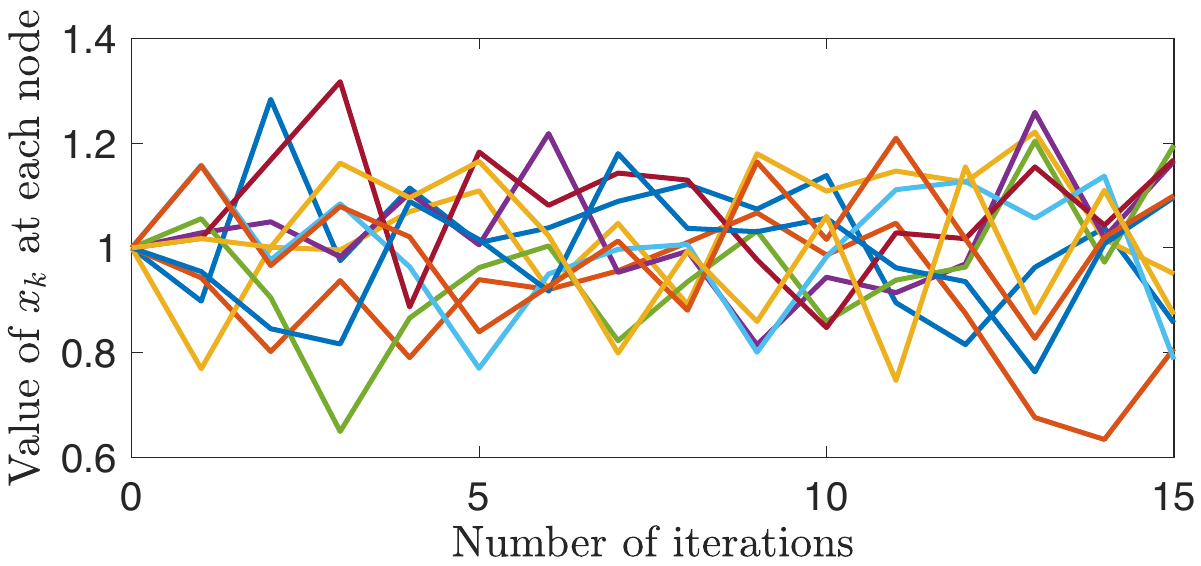}
\vspace{-0.03cm}
\includegraphics[width=0.95\columnwidth]{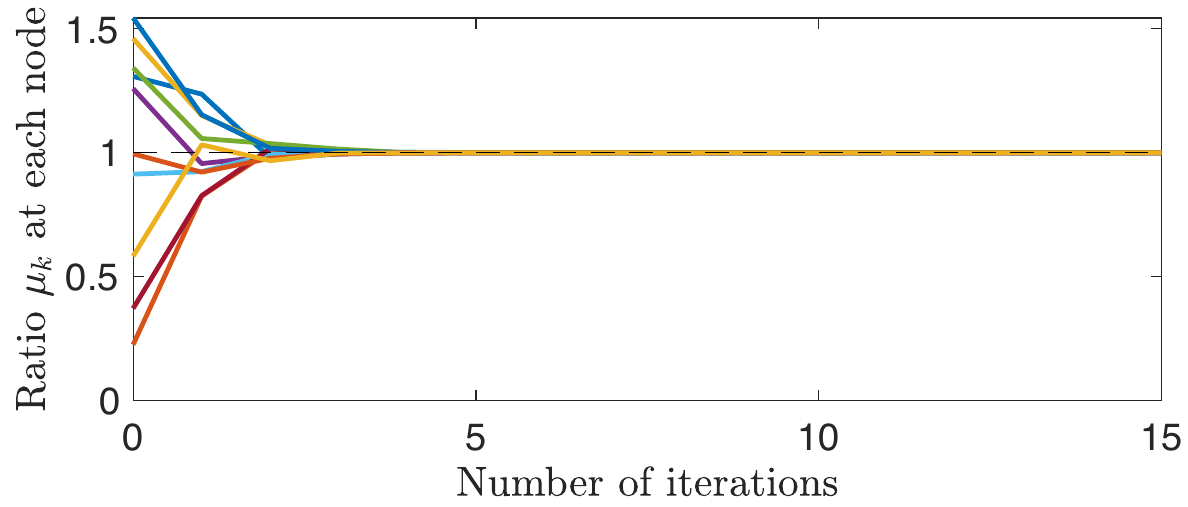}
\caption{Time-varying channel Over-the-Air Ratio Consensus for a graph of $10$ nodes. The individual iterations do not converge, but the ratio does.}
\label{example2}
\end{figure}
%
%
%
%
\section{Conclusions and Future Directions}\label{sec:conclusions}


We considered the average consensus problem in wireless multi-agent systems. 
A distributed protocol employing over-the-air aggregation 
was built upon a modified version of Ratio Consensus, in which an extra normalization step is needed to account for the arbitrary values of the channel coefficients. Numerical simulations corroborate the validity of our theoretical results.
%
%

Part of our ongoing work is to investigate how our method can be modified to account for channel noise and non-coherent transmissions, without requiring any knowledge about the noise statistics and without having a diminishing step size.

\bibliographystyle{IEEEtran}
\bibliography{references}

\begin{thebibliography}{10}
\providecommand{\url}[1]{#1}
\csname url@samestyle\endcsname
\providecommand{\newblock}{\relax}
\providecommand{\bibinfo}[2]{#2}
\providecommand{\BIBentrySTDinterwordspacing}{\spaceskip=0pt\relax}
\providecommand{\BIBentryALTinterwordstretchfactor}{4}
\providecommand{\BIBentryALTinterwordspacing}{\spaceskip=\fontdimen2\font plus
\BIBentryALTinterwordstretchfactor\fontdimen3\font minus
  \fontdimen4\font\relax}
\providecommand{\BIBforeignlanguage}[2]{{%
\expandafter\ifx\csname l@#1\endcsname\relax
\typeout{** WARNING: IEEEtran.bst: No hyphenation pattern has been}%
\typeout{** loaded for the language `#1'. Using the pattern for}%
\typeout{** the default language instead.}%
\else
\language=\csname l@#1\endcsname
\fi
#2}}
\providecommand{\BIBdecl}{\relax}
\BIBdecl

\bibitem{2003:jadbabaie_coordination}
A.~Jadbabaie, J.~Lin, and A.~Morse, ``Coordination of groups of mobile
  autonomous agents using nearest neighbor rules,'' \emph{{IEEE} Transactions
  on Automatic Control}, vol.~48, no.~6, pp. 988--1001, June 2003.

\bibitem{nazer2007computation}
B.~Nazer and M.~Gastpar, ``Computation over multiple-access channels,''
  \emph{IEEE Transactions on Information Theory}, vol.~53, no.~10, pp.
  3498--3516, 2007.

\bibitem{Boche:2012WiOpt}
M.~Goldenbaum, H.~Boche, and S.~Stańczak, ``Nomographic gossiping for
  $\mathcal{f}-$consensus,'' in \emph{International Symposium on Modeling and
  Optimization in Mobile, Ad Hoc and Wireless Networks (WiOpt)}, 2012, pp.
  130--137.

\bibitem{2012:WCNC}
M.~Zheng, M.~Goldenbaum, S.~Stańczak, and H.~Yu, ``Fast average consensus in
  clustered wireless sensor networks by superposition gossiping,'' in
  \emph{IEEE Wireless Communications and Networking Conference (WCNC)}, 2012,
  pp. 1982--1987.

\bibitem{2021:Max-Consensus}
F.~Molinari, N.~Agrawal, S.~Stańczak, and J.~Raisch, ``Max-consensus over
  fading wireless channels,'' \emph{IEEE Transactions on Control of Network
  Systems}, vol.~8, no.~2, pp. 791--802, 2021.

\bibitem{2022:Molinari_Max-Consensus}
------, ``Over-the-air max-consensus in clustered networks adopting half-duplex
  communication technology,'' \emph{IEEE Transactions on Control of Network
  Systems}, pp. 1--10, 2022.

\bibitem{2018:Molinari}
F.~Molinari, S.~Stanczak, and J.~Raisch, ``Exploiting the superposition
  property of wireless communication for average consensus problems in
  multi-agent systems,'' in \emph{European Control Conference (ECC)}, 2018, pp.
  1766--1772.

\bibitem{yang2024distributed}
H.~Yang, X.~Chen, L.~Huang, S.~Dey, and L.~Shi, ``Distributed average consensus
  via noisy and non-coherent over-the-air aggregation,'' \emph{arXiv preprint
  arXiv:2403.06920}, 2024.

\bibitem{michelusi2022non}
N.~Michelusi, ``Non-coherent over-the-air decentralized stochastic gradient
  descent,'' \emph{arXiv preprint arXiv:2211.10777}, 2022.

\bibitem{2004:Reciprocity_Smith}
G.~Smith, ``A direct derivation of a single-antenna reciprocity relation for
  the time domain,'' \emph{IEEE Transactions on Antennas and Propagation},
  vol.~52, no.~6, pp. 1568--1577, 2004.

\bibitem{2010:christoforos}
A.~D. {Dom\'{i}nguez-Garc\'{i}a} and C.~N. Hadjicostis, ``Coordination and
  control of distributed energy resources for provision of ancillary
  services,'' in \emph{Proceedings of the First {IEEE} International Conference
  on Smart Grid Communications}, Oct. 2010, pp. 537--542.

\bibitem{1963:Wolfowitz}
J.~Wolfowitz, ``Products of indecomposable, aperiodic, stochastic matrices,''
  \emph{Proceedings of the American Mathematical Society}, vol.~15, pp.
  733--737, 1963.

\end{thebibliography}
%
%
%
%
\end{document}